\documentclass[conference]{IEEEtran}
\usepackage[utf8]{inputenc}
\usepackage[ruled,vlined]{algorithm2e}
\usepackage[english]{babel}
\usepackage{csquotes}
\usepackage{amsthm}
\newtheorem*{remark}{Remark}
\newtheorem{theorem}{Theorem}
\theoremstyle{definition}
\newtheorem{definition}{Definition}[section]
\newtheorem{lemma}[theorem]{Lemma}
\newtheorem{corollary}[theorem]{Corollary}

\usepackage{graphicx}
\graphicspath{ {./images/} }
\usepackage{mathtools}

\usepackage{cite}
\usepackage{amsmath,amssymb,amsfonts}
\usepackage{algorithmic}
\usepackage{textcomp}
\usepackage{xcolor}
\def\BibTeX{{\rm B\kern-.05em{\sc i\kern-.025em b}\kern-.08em
    T\kern-.1667em\lower.7ex\hbox{E}\kern-.125emX}}

\begin{document}

\title{An Analysis of Multi-hop Iterative Approximate Byzantine Consensus with Local Communication}

\author{
\IEEEauthorblockN{Matthew Ding}
\IEEEauthorblockA{
\textit{Westford Academy, MIT PRIMES}\\
Westford, USA \\
matthewding@berkeley.edu}
}

\maketitle
\begin{abstract}
Iterative Approximate Byzantine Consensus (IABC) is a fundamental problem of fault-tolerant distributed computing where machines seek to achieve approximate consensus to arbitrary exactness in the presence of Byzantine failures. We present a novel algorithm for this problem, named Relay-IABC, which relies on the usage of a multi-hop relayed messaging system and crytographically secure message signatures. The use of signatures and relays allows the strict necessary network conditions of traditional IABC algorithms to be circumvented. In addition, we show evidence that Relay-IABC achieves faster convergence than traditional algorithms even under these strict network conditions with both theoretical analysis and experimental results.
\end{abstract}

\section{Introduction}
The idea of the Byzantine fault-tolerance problem was first introduced in the seminal paper of Lamport et al. \cite{ByzantineGeneral}. Byzantine consensus has since become a large research topic, with applications in a variety of fields including blockchain technology and machine learning \cite{Blockchain} \cite{ByzantineML}. 

Dolev et al. \cite{Dolev} modified and extended the problem of Byzantine agreement by introducing \textit{approximate Byzantine agreement}, allowing machines to reach approximate consensus rather than exact consensus. This was motivated by the fact that exact consensus in asynchronous systems was proven to be impossible \cite{Lynch}. Additionally, in synchronous systems, approximate Byzantine consensus can be used to create algorithms that do not require complete knowledge of the network topology \cite{IABCPart1}.

This approximate Byzantine consensus problem aims to have all honest machines converge to a single state within the convex hull of initial states as the number of iterations approaches infinity \cite{Dolev}. Vaidya \cite{Transition} utilizes a method describing the progression of states in the network using transition matrices to prove consensus.

The main contribution of this paper is to utilize two key tools that have seen a lot of use and success in traditional Byzantine consensus problems: crytographically secure signatures \cite{DolevStrong, ByzantineGeneral,Dishonest} and message relays \cite{ByzantineGeneral,Relay}; our work is a generalization of the work presented in \cite{Transition}, with signatures and relays used to circumvent certain network assumptions that would otherwise be necessary.

Much work has been done on the application of multi-hop communications for standard Byzantine Broadcast \cite{Maurer, Drabkin}. Additionally, the work done in \cite{Su} extended multi-hop communication to the setting of IABC problems. However, our work is significant in two ways: First, it differs from \cite{Su} in that it considers the case of only local communication, where messages may only be sent directly to direct neighbors as opposed to more distant ones. Second, our work also analyzes the convergence rate of IABC algorithms, which has not been done previously in the multi-hop communication setting. 

Byzantine consensus has had applications to machine learning by having machines perform gradient descent steps to minimize a loss function on top of consensus techniques. The combination of delays and signatures may have additional applications in Byzantine gradient descent. \cite{BRIDGE} extends approximate Byzantine consensus algorithms for use in Byzantine gradient descent methods. Our work analyzes the convergence rate of IABC algorithms, which is of interest in any application to machine learning protocols.

\section{Problem Formulation}
\subsection{Definitions}
\subsubsection{Network Definitions}
\begin{itemize}
    \item Let $m$ be the total number of machines (or "nodes"), $h$ the number of honest machines, and $b$ be the number of Byzantine machines ($h+b=m$). 
    
    \item Let $B$ denote the set of byzantine nodes and $H$ denote the set of honest nodes. The set of all nodes $V$ is equal to $B \cup H$.
    
    \item Denote $N_{i}^{I}$ to be the set of all machines that have incoming edges from machine $i$. Denote $N_{i}^{O}$ to be the set of all machines that have outgoing edges to machine $i$
    
    \item Define $dist(i,j)$, for some $i,j \in V$ as the length of the shortest path from node $i$ to $j$
\end{itemize}

\subsubsection{Matrix Definitions}
    \begin{itemize}
        \item Throughout this paper, we will refer to a value that can be lower bounded by some arbitrary positive constant as a "non-zero value".
        
        \item A non-zero column denotes a column of a matrix that is filled entirely with non-zero values.
    
        \item Transition matrix $M$ denotes a square matrix of size $h\times h$
        
        \item $M_{i}[t]$ denotes the $i$th row of matrix $M[t]$
        
        \item $M_{ij}[t]$ denotes the element at row $i$ and column $j$ of matrix $M[t]$
        
        \item We define the first row and column in every matrix as row/column 0
    \end{itemize}

\subsection{Decentralized Communication Model}
We consider a static, directed network $G(V, E)$, where $V = \{0,1,2,..., m-1\}$ and $E$ representing communication links between neighboring nodes. If $(i, j) \in E$, then node $i$ may send messages to node $j$. 

Our protocol is analyzed in the synchronous communication setting, where communication occurs over a sequence of iterations. Messages sent during an iteration are guaranteed to be received by the intended recipient before the beginning of the subsequent iteration, and given finite amount of time. 

Each node $i\in V$ starts with an initial real-valued input. The goal of the protocol is to approach a state that satisfies the following two conditions:
\begin{enumerate}
    \item Validity condition: At the beginning of each iteration of the protocol, the state of each honest node remains within the convex hull of the states of all honest nodes at the beginning of the previous iteration.
    
    \item Convergence condition: The difference between the states of any two honest nodes approaches zero as the number of iterations approaches infinity.
\end{enumerate}

To make analysis easier, we will differentiate between \textit{iterations} \textit{phases}. We define an \textit{iteration} as a given finite duration of time where machines may communicate with one another. During each iteration, every honest node sends and receives messages from its neighbors. Define a \textit{phase} as a set of $D$ iterations. Therefore the $ith$ phase contains iterations $iD$ to $(i+1)D$. Since the distance between any two honest nodes is at most $D$, any message from an honest node is guaranteed to reach all other honest nodes within $D$ iterations.

Note that even though nodes do not have edges connecting to themselves, they may always "send" messages to themselves.

\subsection{Byzantine Failure Model}
Among the $m$ machines in the decentralized network, $b$ of them are Byzantine machines. Byzantine machines may deviate arbitrarily from the standard protocol. For example, a Byzantine machine may output any arbitrary real value, and send mismatching messages to each of its neighbors. However, key restrictions of Byzantine nodes are that they cannot forge signatures of honest users, and they cannot manipulate the contents of existing signed messages.

\section{Relay-IABC Algorithm}
\subsection{Assumptions}
\theoremstyle{definition}
\begin{definition}{Honest Subgraph:}
Define the honest subgraph as the graph that is formed by removing all byzantine nodes and all edges connected to byzantine nodes in the original graph.
\end{definition}
\begin{itemize}
    \item The number of byzantine machines is strictly less than one-third the total number of machines ($b<\frac{1}{3}m$).
    
    \item We assume the honest subgraph is bidirectionally connected (there exists directed paths from every honest node to every other honest node).
    
    \item We assume the diameter of the honest subgraph is upper-bounded by $D$. 
\end{itemize}

Note that individuals machines do not need to have knowledge of the exact network topology besides the value of $D$.

\subsection{Our Contributions}
Our work is an extension of the work done in \cite{Transition}. Our network assumptions are much less restrictive. In particular, we assume no network connectivity assumptions besides the honest subgraph being bidirectionally connected.

The goal of this protocol is to use a relay system to bypass traditional network connectivity assumptions outlined in \cite{IABCPart2}, which has a necessary but insufficient condition that each node has an indegree of at least $2b+1$. This requires that each honest node have at least $b+1$ incoming edges from honest neighbors. On the other hand, bidirectional connectivity of the honest subgraph may possibly be achieved when each honest node has a maximum indegree of as low as one honest neighbor. 

One additional advantage of Relay-IABC is that even for graphs with the strict network assumptions of \cite{Transition}, we show evidence that our algorithm achieves faster convergence than traditional non-relay algorithms both theoretically and empirically.

This comes at the tradeoff of higher communication costs, as now machines send to each other at most $m$ sets of parameters in each message, as opposed to one parameter in most other algorithms in literature. 

Our paper proposes an Iterative Approximate Byzantine Consensus algorithm with unforgeable signatures, which to our knowledge has never been done before. Signatures have seen lots of successful usage in standard byzantine consensus, but until now have not been used in Iterative Approximate Byzantine Consensus methods.

\subsection{High-Level Idea}
Since the honest subgraph is bidirectionally connected, this allows all honest machines to receive signed messages from every other honest machine through a broadcast and relay system. Thus we create a pseudo-complete communication graph over the course of an entire phase of $D$ rounds. Since a complete graph does indeed satisfy the necessary conditions of \cite{IABCPart2}, our relay protocol achieves convergence as well.

We use a trimmed-mean aggregation step \cite{BRIDGE,Transition} to ensure Byzantine robustness. The trimmed-mean step works by eliminating the greatest $b$ values and the smallest $b$ values, and then taking the arithmetic mean of the remaining values. By removing the greatest and least $b$ values, we ensure that the maximum and minimum values of the set of remaining values are both values of honest machines. This prevents Byzantine machines from forcing the states of honest machines to deviate arbitrarily.
\medskip

Each honest machine $i$ keeps track of their own vector $v_{i}$. This vector consists of $(v_{i}(0), v_{i}(1)... v_{i}(m-1))$, where $v_{i}(j)$ represents machine $i$'s most recently updated record of machine $j$'s state. $v_{i}$ may not always contain a state that was received from an actual message for each machine in the network, as machine $i$ may not have yet have received a message from all machines within the current iteration. 

Machine $i$ only performs a trimmed-mean step after each $D$ iterations, as this ensures that each broadcast of all honest machines has had sufficient time to relay across the entire graph and reach every other honest node. This means that each honest machine is outputting an identical message, their current state value, for $D$ consecutive iterations.

We choose the specific value $D$, the upper-bound of the diameter of the honest subgraph, so after $D$ iterations, all honest vectors will always contain more honest parameters than byzantine parameters. $D$ is in the worst-case $O(h)$, but with high probability is $O(\log h)$ in Erdos-Renyi random graphs \cite{RandomGraph}.

\medskip

The Relay-IABC algorithm guarantees that for all honest machines $i$ and $j$, any state $v_i(j)$ in the vector $v_i$ will contain a valid signature from machine $j$.

\subsection{Relay-IABC Algorithm} 

\begin{algorithm} \label{Alg1}
\SetAlgoLined
    \begin{remark}
        This algorithm is implemented by a specific honest machine $i$. Each honest machine $i \in H$ will implement this algorithm concurrently.
    \end{remark}
    
    \KwResult{Each state $v_{i}(i)$ remains within the convex hull of the initial states at each Iteration $t$, and each state converges to the same value as Iteration $t \rightarrow \infty$.}

    Initialization:
 
    $v_{i}(i)\gets$ Initial State of node $i$ (with signature $i$).
    
    \medskip
 
    \For{Iteration $t\gets0$ \KwTo $T$}{
        Broadcast $v_i$ to all machines $j\in N_i^O$
        
        Receive $v_j$ from all machines $j\in N_i^I$
        
        \begin{remark}
            When receiving $v_{j}$, ignore all parameters received that are not properly signed. If no proper message is received from a certain node, set their incoming value to be an arbitrary predefined real value (e.g. 0).
        \end{remark}

        $G_{i}\gets N_i^O \cup \{i\}$
        
        \medskip
        
        \For{$j\gets0$ \KwTo $m-1$}{
            \begin{remark}
                In the next two lines, we do the following: Out of all parameters $v(j)$ received from the broadcast step, set $v_{i}(j)$ to a single arbitrary one $v'(j)$
            \end{remark}
            
            \If{$j \neq i$}{
                $v_i(j)\gets v'(j)$
            }
        }
        
        \medskip
       
        \If{$t \bmod D = 0$}{
            \textbf{Trimmed-mean update step:} 
            
            \medskip
            In a new vector, sort the values of $v_{i}$ in increasing order:
            
            \begin{equation}
                v^*_{i} \gets sort(v_{i})
            \end{equation}
            
            Ignore the least and greatest $b$ values, and set the value of $v_{i}(i)$ to be the average of all remaining values in $v^*_{i}$, as defined below:
            
            \begin{equation} \label{update}
                v_{i}(i) \gets \frac{1}{m-2b} \sum_{k=b}^{m-b-1} v^*_{i}(k)
            \end{equation}
       
            Add signature $i$ to $v_{i}(i)$
        }
    }
 \caption{Relay-IABC}
\end{algorithm}

In the following two sections, we prove the correctness of Relay-IABC by showing that it satisfies both the validity and convergence conditions.

\section{Validity of Relay-IABC}
Let us consider an arbitrary iteration $t$ of the Relay-IABC protocol such that $t\geq 1$. We prove the following theorem.

\begin{theorem}
    For each honest node $h$, the state of $h$ at the end of iteration $t$ remains within the convex hull of states of honest nodes from the end of iteration $t-1$.  
\end{theorem}

\begin{proof}
    In iteration $t$, each honest node $h$ receives every single honest nodes state from iteration $t-1$ (including its own), along with an arbitrary state from each byzantine node in the network. Thus, each honest node receives a total of $2b+1$ honest states and $b$ byzantine states for the trimmed-mean update step in Algorithm \ref{Alg1}. In this step, each honest node will sort their list of states and remove the greatest and least $b$ values each. 
    
    For some arbitrary $h$, without loss of generality, define the set of the lowest $b$ values as $L$, the greatest $b$ values as $G$, and the remaining $b+1$ values of $M$. The state of node $h$ at the end of iteration $t$ is the arithmetic mean of all states within $M$. We now consider two possible cases of where byzantine nodes lie within these three sets.
    
    \begin{enumerate}
        \item Case 1: $\lvert B \cap M \rvert = 0$
        
        If all states within $M$ are from honest nodes, this guarantees that the arithmetic mean of all values within set $M$ will be within the convex hull of all honest states from iteration $t-1$, proving the theorem.
        
        \item Case 2: $\lvert B \cap M \rvert \neq 0$
        
        In this case, there exists at least one Byzantine node's state within the set $M$. This means that $\lvert H \cap M \rvert \leq b$ and $\lvert (L \cup G) \cap H \rvert \geq b+1$. Since the size of both $L$ and $G$ is $b$, there exists at least one honest node's state in both $L$ and $G$. This implies that every single Byzantine state within $M$ lies within the convex hull of honest states from iteration $t-1$. The arithmetic mean of all states within $M$ is thus also within this convex hull, proving the theorem.
    \end{enumerate}
\end{proof}

\section{Convergence of Relay-IABC}
\subsection{Overview} \label{overview}
In this section, we prove that the Relay-IABC algorithm satisfies the convergence condition for IABC. Our algorithm and proof is similar to that which is proposed in \cite{Transition}. We will refer to said algorithm as "IABC", and the novel algorithm proposed in this paper as "Relay-IABC". We utilize the same transition matrices $M[t]$ of size $h\times h$ in \cite{Transition} to model the network iterations.

The main difference between the two algorithms is that each single transition matrix in Relay-IABC corresponds a single phase, or a set of a $D$ consecutive transition matrices in IABC. Since each node will receive a value from every other node in the graph within $D$ iterations, communication in Relay-IABC per $D$ iterations can be represented as a single iteration within a complete graph. In the next subsection, we prove that a complete graph satisfies the network connectivity assumptions of IABC. Thus, Relay-IABC satisfies the condition of convergence for IABC algorithms \cite{Transition}.

\subsection{Source Component Proof} \label{SourceProof}
We define a source component as an honest node that has a directed path of finite length to all other honest nodes in the network. A necessary condition outlined in \cite{Transition} for a specific network to be able to converge with an IABC algorithm is that any arbitrary "reduced graph" (see Definition \ref{reduced}) of the network must contain a source component.

\theoremstyle{definition}
\begin{definition}{Complete Graph:}
A complete graph is a graph with vertex set $V$, and edge set $E'$, such that $\forall i,j$ such that $i \neq j: (i,j) \in E'$. The network itself contains $b$ Byzantine nodes, $h$ honest nodes, and $\|V\| = m$.
\end{definition}

A complete graph describes the de-facto communication during an entire phase: since the longest path between any two honest nodes is at most $D$, any two nodes may communicate with each other for at least one iteration during every single phase. We now introduce a graph that represents the network graph after the trimming in (\ref{update}).

\theoremstyle{definition} 
\begin{definition}{Reduced Graph:} \label{reduced}
A general reduced graph is a specific graph with all nodes in set $B$ removed, along with their incoming and outgoing edges. Additional, we remove any arbitrary set of $b$ incoming edges from each remaining node. 
\end{definition}

From now on, we choose to use the term "reduced graph" to specifically refer to the reduced graph of the complete graph in this paper. The reduced graph represents the "de-facto communication links" between honest nodes after the trimmed-mean step is completed.

Note that there may be multiple reduced graphs for every complete graph, but only a finite number of them. Define $R_{f}$ to be the set of all reduced graphs for a given complete graph, and define $\tau$ as $\| R_{f} \|$. Note that this definition comes from \cite{Transition}.
\medskip

We define a source component of a graph as an honest node that has a directed path to every other honest node in a graph. The following lemma is necessary to prove convergence of Relay-IABC \cite{Transition}:

 \begin{lemma} 
    Any arbitrary reduced graph contains at least one source component.
 \end{lemma}
 
 \begin{proof} 
     A reduced graph is constructed by removing $n$ incoming edges from each node of a fully connected directed graph of at least $2n+1$ nodes. In this proof, we assume that the reduced graph has exactly $2n+1$ nodes. The case where the number of nodes exceeds $2n+1$ is a simple generalization.

    In a fully connected graph of $2n+1$ nodes, there are totally (2n+1)2n outgoing edges. Thus, after removing $n$ incoming edges (which are also outgoing edges of some other arbitrary node) from each node, there still exists at least $(2n+1)n$ outgoing edges left in the graph. By Pigeonhole Principle, at least one node in the reduced graph, let us denote it as $v_0$, has at least $n$ outgoing edges. 
    
    Denote the set of all nodes with direct incoming edges from $v_0$ as set $S$. We know that $\|S\|\geq n$. For each of the nodes which are not in the set $S$ (there are no more than n nodes not in S), it is noted that each of them has $n$ incoming edges. 
    
    Assume that none of these nodes have incoming edges from $v_0$ or any node in set $S$. Thus, they can only have edges from at most a total of $2n+1 - 2 - n = n-1$ nodes. However, it is known that all nodes have $n$ incoming edges. This is a contradiction. Thus, they must either have one incoming edge from a node in $S$, or an incoming edge from $v_0$. Therefore we have proven that $v_0$ is a source component.
\end{proof}

The above proof also proves the following corollary.
\begin{corollary} \label{SourceCorollary}
    At least one node in a reduced graph of size $2n+1$ contains at least $n$ outgoing edges.
\end{corollary}

This corollary will be used later in the paper to derive the convergence rate of the Relay-IABC algorithm.

\subsection{Transition Matrix Analysis} \label{Theoretical}
As explained in Section \ref{overview}, we use transition matrices to model the iteration of network states. However, each matrix $M'[t]$ models communication over a set of $D$ consecutive iterations (one phase) as opposed to just a single iteration. An intuitive understanding of the algorithm is that the communication graph per transition matrix $M'[t]$ is a complete graph before trimming and a reduced graph after trimming.

\medskip
The detailed proof of convergence is shown below:

 \begin{lemma} \label{L1}
    The product of two transition matrices $M^1[t] \times M^2[t]$ will have a column with at least $b+1$ non-zero values.
 \end{lemma}
 
 \begin{proof}
     From Corollary \ref{SourceCorollary}, it is shown that there exists a single node within the network with outgoing communication edges to at least $b$ different honest neighbors \textit{after the trimmed-mean step}. This implies that this single node has its values received by at least $b+1$ nodes (including itself). Let us call this central node Node $j$. Thus, at least $b+1$ rows will have a non-zero value in column $j$, proving the lemma.
 \end{proof}
 
  \begin{lemma} \label{L2}
    Every transition matrix row $M_i[t]$, for all $0\leq i \leq h$, will contain exactly $b+1$ non-zero values.
 \end{lemma}
 
 \begin{proof}
     A transition matrix models the communication during a single phase, or a set of D iterations. This is modeled by a complete graph, so each node receives exactly $3b+1$ values from unique nodes. The trimming step will remove exactly $2b$ of these, leaving exactly $b+1$ remaining non-zero values $M_i[t]$
 \end{proof}
 
 Lastly, we introduce one more transition matrix definition:
 \theoremstyle{definition} 
\begin{definition}{Scrambling Matrix:}
A scrambling matrix is defined by a transition matrix with at least one non-zero column.
\end{definition}

 We now introduce our main result: 
\begin{theorem} \label{T1}
     The product of three transition matrices $M^1[t] \times M^2[t] \times M^3[t]$ will result in a scrambling matrix.
\end{theorem}

\begin{proof}
    From Lemma \ref{L1}, $M^1[t] \times M^2[t]$ contains a column of at least $b+1$ non-zero values. With loss of generality, define this column as Column $j$. From Lemma \ref{L2}, $\forall i$ such that $0\leq i < 2b+1$, $M^3_i[t]$ contains at least $b+1$ non-zero values.
    
    The size of any transition matrix $M[t]$ is $(2b+1) \times (2b+1)$. From the Pigeonhole Principle, $\forall i$ such that $0\leq i < 2b+1$, $\exists z$ such that $(M^1[t] \times M^2[t])_{iz}$ and $M^3_{zj}$ are both non-zero values. This implies that $\forall i$ such that $0\leq i < 2b+1$, $(M^1[t] \times M^2[t] \times M^3[t])_{ij}$ is a non-zero value. In other words, column $j$ of matrix $M^1[t] \times M^2[t] \times M^3[t]$ is a non-zero column, proving the theorem.
\end{proof}

\medskip
From \cite{Transition}, Theorem \ref{T1} is sufficient to prove that Relay-IABC satisfies the convergence condition.

\section{Convergence Rate Analysis}
In this section we compare the convergence rates of IABC and compare it to Relay-IABC. We once again extend the transition matrices presented by Vaidya in \cite{Transition} to show how Relay-IABC achieves convergence at a faster rate.

\subsection{IABC Analysis}
Let $M[t]$ be the transition matrix that models the update of iteration $t$. We now present some convergence analysis of IABC from \cite{Transition}:

\begin{align}
\lim_{t\rightarrow \infty} \delta(\Pi_{i=1}^t M[t])
\leq \lim_{t\rightarrow\infty} \Pi_{i=1}^t \lambda(M[t]) \\ \label{rate}
\leq \lim_{i\rightarrow\infty} \Pi_{i=1}^{\lfloor\frac{t}{h\tau}\rfloor} \lambda(Q(i)) \\
= 0 
\end{align}

Additionally, 
\begin{align}
    \lambda(M[t])\leq 1 \\
    \lambda(Q(i))\leq (1-\beta^{h\tau})<1 \label{scrambling}
\end{align}

The network converges if and only if $\lim_{t\rightarrow \infty} \delta(\Pi_{i=1}^t M[t])$ \cite{Transition}. Thus, we see the convergence rate honest nodes in the entire network can be described empirically as the rate at which \eqref{rate} approaches 0. The exponent in \eqref{scrambling} comes from the fact that it takes the product of $h\tau$ transition matrices $M[t]$ to form scrambling matrix $Q(i)$ given a network following the network assumptions of \cite{Transition}. 

\subsection{Relay-IABC Analysis}
We now analyze the specific convergence rate of Relay-IABC compared to IABC. From \ref{T1}, we have proven that Relay-IABC only requires the product of three transition matrices to form a scrambling matrix. From \cite{Transition}, it also is shown that:
\begin{align}
\lim_{t\rightarrow \infty} \delta(\Pi_{i=1}^t M'[t])
\leq \lim_{t\rightarrow\infty} \Pi_{i=1}^t \lambda(M'[t]) \\ \label{newrate}
\leq \lim_{i\rightarrow\infty} \Pi_{i=1}^{\lfloor\frac{t}{3D}\rfloor} \lambda(Q'(i)) \\
= 0 
\end{align}

Given that,
\begin{equation}
    \lambda(Q(i))\leq (1-\beta^{3D})<1
\end{equation}

and $3D << h\tau$, then we see that \eqref{newrate} approaches zero much faster than \eqref{rate}. This evidence suggests that Relay-IABC achieves network converges at a significantly faster rate than traditional IABC. In our next section, we support this claim with empirical data.

\section{Practical Applications}
In this section we discuss the merits of the Relay-IABC algorithm in the context of potential applications in the real-world.

\subsection{Convergence Rate}
In this section, we seek to provide empirical evidence for analysis done in Section \ref{Theoretical} through simulations. All code can be found at: https://github.com/matthew-ding/primes-project-2021.

We ran Python scripts to simulate a network of nodes running both the IABC and Relay-IABC algorithms. The network was generated as a random Erdos-Renyi graph ($p=0.8$), with 30 honest nodes and 14 Byzantine nodes. Each honest node was given a random initial state within $(-110, 110)$ and byzantine nodes would output a random value approximately within that range with slight variation depending on which honest node they were communicating with. 

In Figure \ref{fig:graph}, we plot the standard deviation of all honest nodes in the network over the iterations of the algorithm. We see that the Relay-IABC algorithm has an empirically faster convergence rate given the simulation parameters.

\begin{figure} [htbp]
    \centerline{\includegraphics[width=1\columnwidth]{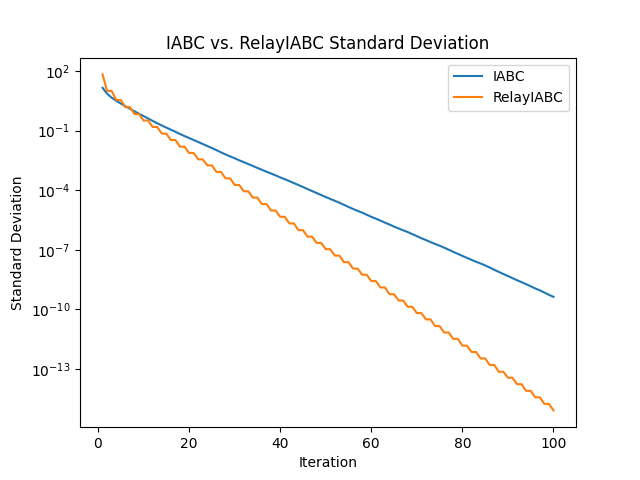}}
    \caption{Comparing convergence rates of IABC and Relay-IABC}
 
    \label{fig:graph}
\end{figure}

\subsection{Network Connectivity}
Besides Relay-IABC having faster empirical convergence rates as compared to IABC, another major benefit is that Relay-IABC may to implemented on sparse networks. On the other hand, IABC requires very dense network connections among honest nodes (necessary condition of $2b+1$ neighbors per node). Functionality on sparse networks is a major requirement for real-life applications, as most graphs that appear in the "real-world" are generally sparse (e.g. computer networks) \cite{Sparse}.

\section{Summary and Future Work}
In this paper, we introduce the Relay-IABC algorithm for iterative approximate Byzantine consensus. The algorithm extends the traditional IABC algorithm \cite{Transition} with the novel usage of signed and relayed messages. We also compare the convergence rates of IABC and Relay-IABC with both theoretical analysis of transition matrices as well as simulation results with Python.

Additionally, we also theorize that the upper-bound of the proportion of Byzantine nodes in this paper of $b < \frac{1}{3}m$ is not tight. It is shown that in approximate consensus algorithms where faulty nodes are not allowed to equivocate (send mismatching messages to separate nodes during the same iteration), algorithms may actually achieve a bound $b < \frac{1}{2}m$ within complete networks \cite{LeBlanc}. We believe that it may be possible to achieve an identical bound for IABC by utilizing cryptography methods to detect Byzantine equivocation. We leave the exact proof for our future work.

Lastly, the communication method of the Relay-IABC algorithm is that of "flooding", where nodes send messages to all their neighbors in order to relay information. While this is a highly robust method of communication, it is also incredibly communication-heavy. The overlay methods proposed in \cite{Drabkin} seek to circumvent this communication cost in the context of Byzantine Broadcast, and similar results might be able to be found for IABC problems as well.

\section*{Acknowledgments}
First of all, I would like to thank MIT and the MIT PRIMES program for giving me this wonderful opportunity to conduct research.

I'd also like to thank Jun Wan and Prof. Lili Su for all of our discussions. Finally, the biggest thanks goes out to my research mentor, Hanshen Xiao, for all of his support and guidance.

\bibliographystyle{IEEEtran}
\bibliography{main.bib}

\end{document}